\documentclass[aip,jmp,10pt,citeautoscript,reprint]{revtex4-1}

\usepackage{amsmath,amsthm,amsfonts,amssymb,times,bbm,graphicx}

\usepackage{hyperref}
\newtheorem{theorem}{Theorem}

\newtheorem{lemma}[theorem]{Lemma}
\newtheorem{corollary}[theorem]{Corollary}

\newcommand{\ket}[1]{\mbox{$\lvert #1 \rangle$}}

\def\ZZ{\mathbbm{Z}}
\def\RR{\mathbbm{R}}
\def\CC{\mathbbm{C}}
\def\FF{\mathbbm{F}}

\def\Id{\mathbbm{1}}

\DeclareMathOperator{\tr}{tr}

\newcommand{\abs}[1]{\lvert#1\rvert}

\usepackage{color}

\begin{document}

\title{Stabilizer information inequalities from phase space distributions}

\author{David Gross}
\affiliation{
	Institute for Physics, University of Freiburg, Rheinstrasse 10,
	79104 Freiburg, Germany
}
\email{www.qc.uni-freiburg.de}

\author{Michael Walter}
\affiliation{
  Institute for Theoretical Physics, ETH Zurich, Wolfgang--Pauli--Str.~27, 8093 Zurich, Switzerland
}
\email{mwalter@itp.phys.ethz.ch}

\date{February 27, 2013; revised: June 13, 2013}

\begin{abstract}
	The Shannon entropy of a collection of random variables is
	subject to a number of constraints, the best-known examples being
	monotonicity and strong subadditivity.  It remains an open question
	to decide which of these ``laws of information theory'' are also
	respected by the von Neumann entropy of many-body quantum states.
	In this article, we consider a toy version of this difficult problem by
	analyzing the von Neumann entropy of stabilizer states.  We find
	that the von Neumann entropy of stabilizer states satisfies all
	\emph{balanced} information inequalities that hold in the classical
	case. Our argument is built on the fact that stabilizer states have a
	classical model, provided by the discrete Wigner function: The phase-space
	entropy of the Wigner function corresponds directly to the von
	Neumann entropy of the state, which allows us to reduce to the
	classical case.  Our result has a natural counterpart for multi-mode
	Gaussian states, which sheds some light on the general properties of
	the construction. We also discuss the relation of our results to
	recent work by Linden, Ruskai, and Winter \cite{linden_quantum_2013}.
\end{abstract}

%\pacs{}

\maketitle

\section{Introduction and Results}

The Shannon entropy of a discrete random variable $X$ is given by $H(X) = -\sum_x p_x \log p_x$, where $p_x$ is the probability that $X = x$.
Given a collection of random variables $X_1, \ldots, X_n$, we can consider the joint entropy $H(X_I)$ of any non-empty subset $X_I = (X_i)_{i \in I}$ of the variables.
These entropies are not independent---they are subject to a number of linear homogeneous inequalities, known as \emph{information inequalities}, or as the ``laws of information theory'' \cite{pippenger86}.
Conversely, the set of all such inequalities determines the set of
possible joint entropies $(H(X_I))$ up to closure \cite{yeung97}.
There are two classes of \emph{fundamental} or \emph{Shannon-type}
inequalities.
The first is \emph{monotonicity},
stating that the entropy does not decrease if more
random variables are taken into account:
$H(X_{I\cup J}) - H(X_I) \geq 0$.
The second class is
\emph{strong subadditivity},
\begin{equation*}
	H(X_{I}) + H(X_J) - H(X_{I\cap J}) - H(X_{I \cup J}) \geq 0.
\end{equation*}
Since the seminal work of Zhang and Yeung it is known that there are other entropy inequalities which are not implied by those of Shannon type \cite{zhangyeung98}.
In fact, there are infinitely many independent such inequalities \cite{matus07}.
%it is known that $\bar\Gamma^*_n$ is not polyhedral

In quantum mechanics, the state of a quantum state of $n$ particles is described by a density operator $\rho$ on a tensor-product Hilbert space.
The state of any subset $I \subseteq \{1,\ldots,n\}$ of the particles is described by the reduced state $\rho_I = \tr_{I^c} \rho$ formed by tracing out the Hilbert space of the other particles. The natural analogue of the Shannon entropy is the von Neumann entropy $S(\rho) = -\tr \rho \log \rho$ \cite{nielsenchuang04}, and it is of fundamental interest to determine the linear inequalities satisfied by the entropies $S(\rho_I)$ of subsystems \cite{pippenger03}.
The most immediate difference to the classical case is that the von Neumann entropy is no longer monotonic:
global quantum states can exhibit less entropy than their reductions (a signature of entanglement).
Instead, the von Neumann entropy satisfies \emph{weak monotonicity}:
\begin{equation*}
  S(\rho_{I \cup K}) + S(\rho_{J \cup K}) - S(\rho_I) - S(\rho_J) \geq 0.
\end{equation*}
Strong subadditivity, however, famously remains valid for quantum
entropies \cite{liebruskai73}. It is a major open problem in quantum
information theory to decide whether there are any entropy
inequalities beyond the ones given above (see Refs.~\onlinecite{lindenwinter05,cadneylindenwinter12} for some partial progress, including a class of
so-called \emph{constrained inequalities}).

Strong subadditivity is tight for product
states (resp.~for independent random variables).
An entropy inequality $\sum_I \nu_I S(\rho_I) \geq 0$ has that
property if and only if $\sum_{I\ni i} \nu_I = 0$ for all $i$. Such
an inequality is called \emph{balanced} (also \emph{correlative} \cite{han75}
or \emph{homogeneous} \cite{linden_quantum_2013}), and it has been shown that
the classical entropy cone is determined by the set of balanced information inequalities
together with monotonicity (which is not balanced) \cite{chan03}. Balanced entropy
inequalities will play an important role below.

Of course, entropy inequalities are also relevant in the case of continuous variables (also classically, see e.g.\ Ref.~\onlinecite{chan03}) as well as for other kinds of entropies, e.g.\ R\'{e}nyi entropies \cite{lindenmosonyiwinter12}.

\bigskip

In this work we study the entropy inequalities satisfied by two classes of quantum states---namely, \emph{stabilizer states} and \emph{Gaussian states} (which are the continuous-variable counterpart of the former).
These states are versatile enough to exhibit intrinsically quantum features (such as multi-particle entanglement), but possess enough structure to allow for a concise and computationally efficient description. In both cases, \emph{quantum phase-space methods} have been built around them, and it is this point of view we aim to exploit here.
For $n$ systems of dimension $d$, the phase space is $\ZZ_d^{2n}$, while for $n$ bosonic modes
it is given by $\RR^{2n}$. In both cases, it is the direct sum of the single-particle/single-mode phase spaces.

The starting point for our work is the Wigner function, which for Gaussian states as well as for stabilizer states in odd dimensions $d$ is a bona fide probability distribution on the classical phase space (the case of even $d$ requires some more care, see Theorem~\ref{stabilizer theorem}).
We may thus define random variables $X_1,\ldots,X_n$ on the phase space, jointly distributed according to the Wigner function of the given quantum state $\rho$. Here, $X_i$ denotes the component in the single-particle space of the $i$-th particle or mode.
The random variables $X_1, \dots, X_n$ constitute our \emph{classical model}.
This construction is compatible with reduction: the
marginal probability distribution of a subset of variables $X_I = (X_i)_{i \in I}$ is
given precisely by the Wigner function of the reduced quantum state
$\rho_I$.  Our crucial observation then is that certain quantum entropies
are simple functions of the corresponding
classical entropies. More precisely, we find that
\begin{equation}
\label{general result}
	S_2(\rho_I) = H_2(X_I) - C \abs I,
\end{equation}
where $C > 0$ is a universal constant and where $S_2(\rho) = -\log \tr
\rho^2$ and $H_2(X) =
-\log \sum_x p_x^2$ denote the quantum and classical R\'{e}nyi-2
entropy, respectively. In the case of continuous variables, we use the
differential R\'{e}nyi entropy $H_2(X) = -\log \int p_x^2 dx$.
Therefore, if $\sum_I \nu_I H_2(X_I) \geq 0$ is a \emph{balanced}
entropy inequality satisfied by the random variables $X_I$
then the same inequality is satisfied by the quantum state,
\begin{equation*}
	\sum_I \nu_I S_2(\rho_I) = \sum_I \nu_I H_2(X_I) - C \underbrace{\sum_I \nu_I \abs I}_{=0} \geq 0.
\end{equation*}

In the case of stabilizer states (Section~\ref{stabilizer states}), all reduced states $\rho_I$ are normalized projectors (onto the corresponding code subspace), while the $X_i$ are uniformly distributed (on their support).
Thus all R\'{e}nyi entropies agree with each other, and also with the Shannon and von Neumann entropy, respectively:
\begin{equation}
	\label{stab entropies}
	S(\rho_I) = H(X_I) - \abs I.
\end{equation}
As above, it follows that any balanced entropy inequality that is
valid for the Shannon entropies of the $X_I$ is also valid for the
von Neumann entropies of the stabilizer states $\rho_I$.  In
particular, \emph{stabilizer states respect all balanced information
inequalities}, such as the inequalities of non-Shannon type found in
Ref.~\onlinecite{zhangyeung98}.
What is more, our construction can also be understood in the group-theoretical
framework of Ref.~\onlinecite{chanyeung02}.
Here it is well-known that there are inequalities which do not hold for arbitrary
random variables, but only for random variables constructed from certain classes of
subgroups, e.g.~normal subgroups \cite{lichong07}.
Since phase spaces are Abelian groups, it follows that the von Neumann entropies
of stabilizer states also respect such information laws, e.g.\
the \emph{Ingleton inequality} \cite{lichong07}, which is the balanced inequality
\begin{equation}
	\label{ingleton}
	I_\rho(I:J|K) + I_\rho(I:J|L) + I_\rho(K:L) - I_\rho(I:J) \geq 0.
\end{equation}
Here, $I_\rho(I:J) = S(\rho_I) + S(\rho_J) - S(\rho_{I \cup J})$ and $I_\rho(I:J|K) = S(\rho_{I \cup K}) + S(\rho_{J \cup K}) - S(\rho_K) - S(\rho_{IJK})$ are the quantum (conditional) mutual information.

We find it instructive to understand how the above classical model
manages to respect monotonicity, while the quantum state may violate it.
For example, since stabilizer states can be entangled (even maximally
so),
$H(\rho_1) = H(\rho_2) = 1$ and	$H(\rho_{12})=0$ are perfectly valid entropies
of a stabilizer state which obviously violate monotonicity.
Equation~\eqref{stab entropies} states that the classical
model is \emph{more highly mixed} than the quantum one, in the sense
that the entropy associated with a subset $I$ is higher by an amount
of $\abs I$. % dits
That is precisely the amount by which quantum mechanics can violate
monotonicity.

\medskip

In the case of Gaussian states (Section~\ref{gaussian states}), the random variables $X_1,\ldots,X_n$ have a multivariate normal distribution, and we show that the differential R\'{e}nyi-2 entropy in \eqref{general result} can be replaced by the R\'{e}nyi-$\alpha$ entropy for arbitrary positive $\alpha \neq 1$:
\begin{equation*}
	\label{gaussian entropies}
	S_2(\rho_I) = H_\alpha(X_I) - \abs I \left( \log \pi -  \frac {\log \alpha} {1-\alpha} \right),
\end{equation*}
In the limiting case $\alpha \rightarrow 1$, we recover a formula involving the differential Shannon entropy
which has previously appeared in Ref.~\onlinecite{adessogirolamiserafini12}, attributed to Stratonovich.
Thus, \emph{R\'{e}nyi-2 entropies of Gaussian states respect all
balanced information inequalities} that hold for multivariate normal
distributions \cite{holtzsturmfels07,shadbakhthassibi11}.
Interestingly, it is not clear whether a similar statement holds for
the von Neumann entropy of the Gaussian state. This is perhaps an
indication that the connection \eqref{stab entropies} between the Shannon and the von Neumann
entropy for stabilizer states is somewhat
coincidental. The comparison with Gaussian states suggests that the
R\'enyi-2 entropies might be the more fundamental quantities in this
context, that merely happen to agree with the von Neumann entropy in the case of stabilizer states.

\medskip

We conclude this section with a few remarks.
Our work uses the classical model provided by the Wigner function as a
tool for proving statements that do not, a priori, seem to be connected
to phase space distributions. This point of view has been employed
before, e.g.\ to construct quantum expanders \cite{grosseisert08},
to establish simulation algorithms \cite{veitch12, mari12, veitch13}, and
for demonstrating the onset of contextuality \cite{howard13}.  It
would be interesting to see further applications.

While it is known that the Wigner function approach cannot be
straight-forwardly translated to non-stabilizer states
\cite{hudson74,gross06,gross07}, our discussion suggests searching for other
maps from quantum states to probability distributions that reproduce
entropies faithfully, up to state-independent additive constants.

In order to establish the Ingleton inequality \eqref{ingleton}, we have used
the group-theoretical approach to classical information inequalities \cite{chanyeung02}.
It would be highly desirable to find a quantum-mechanical analogue of this work
(see Refs.~\onlinecite{christandlmitchison06,christandlsahinogluwalter12} for partial results towards this goal,
motivated by the quantum marginal problem of quantum physics).

\medskip

{\bf Related Work.} Independently of this work, Linden, Ruskai, and Winter have published an analysis of the entropy cone generated by stabilizer states \cite{linden_quantum_2013}.  Their methods -- focusing on group-theoretical constructions -- are conceptually complementary to our
phase-space approach.
Ref.~\onlinecite{linden_quantum_2013} contains a complete characterization
of the entropy cone generated by four-party stabilizer states.
The paper also
lists further example of inequalities which, like the Ingleton
Inequality, are respected by stabilizer states, even though there are
classical distributions violating it.
While not stated explicitely, their methods can readily be used to
prove that all balanced inequalities remain valid for stabilizers
(see Theorem 11 in
Ref.~\onlinecite{linden_quantum_2013}
and discussion thereafter).

{\bf Convention.} In this work, entropies of $d$-level systems are measured
in units of $\log d$ bits. In the continuous-variable case, we employ the natural
logarithm.

\section{Stabilizer States}
\label{stabilizer states}

In this section, we describe our results on stabilizer states.  We
start by fixing some notation and recalling the basic formalism of
stabilizer states \cite{gottesman96,nielsenchuang04}.  The phase-space
methods that we employ work most smoothly over Hilbert spaces of odd
dimension and therefore the exposition is focused on that case.
However, discrete phase spaces and stabilizer states can be defined for
any dimension $d$ and our main result is valid in full generality.
%Theorem~\ref{stabilizers} summarizes the precise statements on phase spaces that we require below. A proof that also covers the even-dimensional case is presented in Appendix~\ref{stabilizers}.
Theorem~\ref{stabilizers} summarizes the precise statements that we require to prove our results, and we present a self-contained account of the general theory in Appendix~\ref{appendix}.
We then establish our main result -- a classical model for the von Neumann
entropy of stabilizer states -- in Theorem~\ref{stabilizer theorem}.

Let $d>1$ be an odd integer. The discrete \emph{configuration space} of a single particle is
$\ZZ_d$, where $\ZZ_d=\ZZ/d \ZZ$ is the set of congruence classes modulo
$d$.
The associated \emph{phase space} is the $\ZZ_d$-module $\ZZ_d^2 = \ZZ_d \oplus \ZZ_d$.
We denote the components of ``vectors'' $v \in \ZZ_d^2$ by $(p,q)$ in order to emphasize the analogy with momentum
and position in the continuous-variable theory. A crucial piece of
structure is the \emph{symplectic form} defined on the phase space.
It maps vectors $v=(p,q)$ and $v'=(p',q')$ to $[v, v'] = p q' - q p'$.
For each vector $v=(p,q)$, we define a \emph{Weyl operator} acting on
the Hilbert space of complex functions on $\ZZ_d$, which can be identified with $\CC^d$:
It is given by $(w(p,q) \psi)(x) = e^{\frac{2\pi\mathrm{i} }{d}(px-2^{-1} pq)} \psi(x-q)$, where $2^{-1} := (d^2+1)/2$ denotes a multiplicative inverse of $2$ modulo $d$ (this only exists for odd $d$).
A direct calculation shows that
\begin{equation}\label{eqn:heisenberglaw}
	w(v) w(v')
	= e^{\frac{2 \pi \mathrm{i}}{d} 2^{-1} [v,v']} w(v+v').
	%= e^{\frac{2 \pi \mathrm{i}}{d} [v,v']} w(v') w(v).
\end{equation}
Thus the Weyl operators realize a \emph{projective} or
\emph{twisted} representation of the additive group of the phase
space $\ZZ_d^2$ (it is a faithful representation of the \emph{Heisenberg group}
over $\ZZ_d$, see e.g.\ Ref.~\onlinecite{folland89}).

For $n$ particles, the phase space is the direct sum
$V=\bigoplus_{i=1}^n V_i=\ZZ^{2n}_d$
of the single-particle phase spaces $V_i = \ZZ^2_d$. It can be represented on $(\CC^d)^{\otimes n}$ by the tensor product of the single-particle representations, $w(v) = \bigotimes_{i=1}^n w(v_i)$, and the composition law \eqref{eqn:heisenberglaw} continues to hold if we extend the symplectic form linearly.

Let us now consider an \emph{isotropic submodule} $M \subseteq V$,
i.e.\ a submodule on which the symplectic form vanishes.
Isotropicity implies by \eqref{eqn:heisenberglaw} that the
Weyl operators $\{w(m) : m\in M\}$ form an Abelian group---a \emph{stabilizer group}.
%Using the basic relations $w(v)^\dagger = w(-v)$ and $\tr w(v) = d^n \delta_{v,0}$,
One easily verifies that
\begin{equation}\label{eqn:codes}
	P(M):= \frac1{\lvert{}M\rvert}
	\sum_{m\in M} w(m)
\end{equation}
defines an orthogonal projection onto a $d^n / \abs M$-dimensional subspace of $(\CC^d)^{\otimes n}$.
This subspace is called the \emph{stabilizer code} associated with $M$; it contains the vectors in $(\CC^d)^{\otimes n}$ that are invariant under the stabilizer group.
The corresponding \emph{stabilizer state} is the density operator
\begin{equation}
	\label{stabilizer state}
	\rho(M)
	= \frac 1 {\tr P(M)} P(M)
	= \frac 1 {d^n} \sum_{m \in M} w(m),
\end{equation}
and its von Neumann entropy is given by
\begin{equation}
\label{odd entropy}
	S(\rho(M)) = n - \log\ \abs M.
\end{equation}

One obtains a larger set of stabilizer codes by including certain phase factors in the sum in \eqref{eqn:codes} \cite{gottesman96,nielsenchuang04,gross06}. However, all stabilizer codes are locally equivalent to one of the form \eqref{eqn:codes}. Since we are only interested in the entropy of $\rho(M)$ and of its reduced density matrices, we may therefore safely restrict to the class of stabilizer states defined above.

One obtains a simple expression for the reduced state $\rho(M)_I$ by using the relation $\tr w(v_i) = d \, \delta_{v_i,0}$. For this, let $V_I := \{ v \in V : v_i = 0 \text{ for } i \notin I\}$ be the phase space of a subset of particles $I \subseteq \{1,\ldots,n\}$, and set $M_I := M \cap V_I$. Then $\rho(M)_I = \rho(M_I)$, i.e.\ the reduced state is the stabilizer state described by the isotropic submodule $M_I\subseteq V_I$.
From \eqref{odd entropy} we find that
\begin{equation*}
	S(\rho(M)_I) = S(\rho(M_I)) = \abs I - \log\ \abs{M_I}.
\end{equation*}

The following theorem summarizes the aspects of the phase-space picture of
stabilizer states that we will use below to establish our main result.
It is stated in such a way as to also apply to the even-dimensional case.
Note that the above discussion essentially proves Theorem~\ref{stabilizers} for odd $d$.
We give a general proof in Appendix~\ref{appendix}.

\newcommand{\stabilizertheorem}[4]{
\begin{#1}[Stabilizers in phase space]
  #2
	Let $V=\bigoplus_{i=1}^n V_i=\ZZ^{2n}_d$ be the phase space for $n$ particles with
	local dimension $d$, where $d > 1$ is an arbitrary integer.
	There is a one-to-one correspondence between isotropic submodules $M \subseteq V$ and equivalence classes $[\rho(M)]$ of stabilizer states on $(\CC^d)^{\otimes n}$ under conjugation with Weyl operators. Moreover,
	\begin{align}
		[\rho(M)_I] &= [\rho(M_I)],
		#3 \\
		S([\rho(M)_I]) &= \abs I - \log\ \abs{M_I}.
		#4
	\end{align}
	If $d$ is odd then there is a canonical element $\rho(M)$ in each equivalence class, given by \eqref{stabilizer state}. It is compatible with reductions, i.e.\ $\rho(M)_I = \rho(M_I)$.
\end{#1}
}
\stabilizertheorem{theorem}{\label{stabilizers}}{\label{reductions stabilizers}}{\label{von neumann entropy stabilizers}}

If the local dimension $d$ is odd, there exists a discrete Wigner function that replicates many properties of its better-known continuous-variable variant \cite{gross06}.
It is the function on phase space defined by
\begin{eqnarray*}
	W_\rho(v)
	&=&
	\frac 1 {d^{2n}} \sum_{v'\in V} e^{-\frac{2\pi\mathrm{i}}{d} 2^{-1} [v,v']} \tr\big(w(v')^\dagger
	\rho\big).
\end{eqnarray*}
The central observation is that in the case of stabilizer states, the Wigner
function $W_{\rho(M)}$ is % formally defines
a \emph{probability distribution on phase space},
i.e.\ it attains only non-negative values and their sum is one.
In fact \cite{gross05,gross06},
\begin{equation}
\label{eqn:phasespacedist}
\begin{aligned}
	W_{\rho(M)}(v)
	%=& d^{-2n} \sum_{v'\in V} e^{-\frac{2\pi\mathrm{i}}{d} \frac{[v,v']}2} \tr\big(w(v')^\dagger \rho(M)) \\
	=&\frac 1 {d^{2n}} \sum_{v'\in V} e^{-\frac{2\pi\mathrm{i}}{d} 2^{-1} [v,v']} \delta_M(v') \\
	=&\frac 1 {d^{2n}} \sum_{v'\in M} e^{-\frac{2\pi\mathrm{i}}{d} 2^{-1} [v,v']} \\
	=&\frac {\abs M} {d^{2n}} \delta_{M^\perp}(v) \\
	=&\frac 1 {\abs{M^\perp}} \delta_{M^\perp}(v),
\end{aligned}
\end{equation}
where we have defined the \emph{symplectic complement} of $M \subseteq V$
by $M^\perp = \{ v \in V : [v,m]=0\quad \forall m \in M \}$, for which $\abs M \abs{M^\perp} = \abs V = d^{2n}$ and $(M^\perp)^\perp = M$.
Thus the Wigner function of the stabilizer state with isotropic submodule $M \subseteq V$ is given by the uniform distribution on $M^\perp \subseteq V$.

We now show that this construction defines a classical model which reproduces the entropies of the given stabilizer state and its reduced states, up to a certain constant. By phrasing the construction solely in terms of the symplectic complement (hence without recourse to the Wigner function), this result can be established for arbitrary local dimension, even or odd:

% This observation was the original motivation for our general construction.

%stabilizer states are the only ones with that property
%\cite{gross_discrete_2006})
%. This translates to the finite-dimensional
%case a well-known theorem due to to Hudson, Soto, and Claverie, which
%states that Gaussian states are the only pure CV states giving rise to
%a non-negative Wigner function).

\begin{theorem}[Classical model for stabilizer states]
	\label{stabilizer theorem}
	Let $V = \bigoplus_{i=1}^n V_i = \ZZ^{2n}_d$ be the	phase space
	for $n$ particles with local dimension $d$, where $d > 1$ is an arbitrary integer.
	Let $\rho$ be a stabilizer state with isotropic submodule $M \subseteq V$,
	and define a random variable $X = (X_1, \ldots, X_n)$ that takes values uniformly in the symplectic
	complement $M^\perp \subseteq V$. Then,
	\begin{equation}
		\label{main eqn}
		S(\rho_I) = H(X_I) - \abs I,
	\end{equation}
	and the same conclusion holds if we replace the Shannon and von Neumann entropy
	by any R\'{e}nyi entropy.

	If $d$ is odd then the above construction can also be obtained by
	interpreting the Wigner function $W_{\rho}$ as the probability
	distribution of the random variable $X$.
\end{theorem}

\begin{proof}
	To prove \eqref{main eqn}, denote by $\pi_I \colon V \rightarrow
	V_I$ the projection onto the phase space of parties $I \subseteq
	\{1,\ldots,n\}$.  It will be convenient to consider $V_I$ as a submodule of $V$ in the natural way.
	To avoid any notational ambiguity, we denote by $X^{\perp_I}$ the symplectic complement of a subspace $X$	taken within $V_I$.

	Observe that
	\begin{equation}
	\label{perp lemma a}
		\pi_I(M^\perp) \subseteq M_I^{\perp_I}.
	\end{equation}
	Indeed, if $v\in M^\perp$ and $m_I \in M_I$, then
	$[\pi_I(v),m_I]=[v,m_I]=0$.
	On the other hand, we find that
	\begin{equation}
	\label{perp lemma b}
		\pi_I(M^\perp)^{\perp_I}
		\subseteq
		M_I.
	\end{equation}
	To see this, consider a vector $v_I \in V_I$ and note that
	if $v_I \perp \pi_I(M^\perp)$ then $v_I \perp M^\perp$, hence $v_I \in M
	\cap V_I = M_I$ since $(M^\perp)^\perp = M$.
	%Finally, since $(\pi_I(M^\perp)^{\perp_I})^{\perp_I} = \pi_I(M^\perp)$,
	We conclude from \eqref{perp lemma a} and \eqref{perp lemma b} that
	\begin{equation}
	\label{perp lemma}
		\pi_I(M^\perp) = M_I^{\perp_I}.
	\end{equation}

	Note that
	$X_I = \pi_I(X)$. Since $\pi_I$ is a group homomorphism, it follows
	that $X_I$ is distributed uniformly on its range, so that
	\begin{align*}
		H(X_I)
		&= \log\ \abs{\pi_I(M^\perp)}
		= \log\ \abs{M_I^{\perp_I}}
		% = \dim M_I^{\perp_I} \\
		= \log \frac {d^{2 \abs I}} {\abs {M_I}} \\
		&= 2 \abs I - \log\ \abs{M_I}
		= \abs I + S(\rho_I),
	\end{align*}
	where we have used \eqref{von neumann entropy stabilizers} in the last step.
	We have thus established \eqref{main eqn}.

	The same result holds if we replace the Shannon and von Neumann
	entropy by R\'{e}nyi entropies. This is because the stabilizer
	states $\rho(M)_I$ are normalized projectors and each random
	variable $X_I$ is distributed uniformly on its range, so that the
	entropies coincide.

	Finally, it is clear from \eqref{eqn:phasespacedist} that for odd
	$d$ the distribution of $X$	coincides with the Wigner function
	$W_{\rho(M)}$ of the stabilizer state.
	It remains to show that the Wigner function $W_{\rho_I}$ of a
	reduced state $\rho_I$ is obtained by marginalizing the full Wigner
	function
	(in other words: the
	quantum and the classical way of reducing to subsystems commute):
	\begin{equation}\label{eqn:margins} W_{\rho_I}(v) = \sum_{w\, :
		\,w_I = v} W_\rho(w)
	\end{equation}
	for all $v \in V_I$. While this can easily be proved in full
	generality from the definition of the Wigner function,
	it is also true that for the special case of stabilizer states,
	Eq.~\eqref{eqn:margins} follows directly from \eqref{perp lemma}.
\end{proof}

\begin{corollary}
\label{stabilizer corollary}
  Stabilizer states satisfy all balanced information inequalities.
  Moreover, they satisfy the Ingleton inequality \eqref{ingleton}.
\end{corollary}
\begin{proof}
	As described in the introduction, the first claim follows immediately from
	\eqref{main eqn}. This is because for any balanced information inequality
	$\sum_I \nu_I H(X_I) \geq 0$ we necessarily have that \cite{shadbakhthassibi11}
	\begin{equation*}
		\sum_I \nu_I \abs I = \sum_I \left( \sum_{i \in I} \nu_I \right) = \sum_i \left( \sum_{I \ni i} \nu_I \right) = 0.
	\end{equation*}
	Hence the correction term in \eqref{main eqn} cancels as we sum over all
	subsystems:
	\begin{equation*}
		\sum_I \nu_I S(\rho_I) = \sum_I \nu_I H(X_I) - \sum_I \nu_I \abs I \geq 0.
	\end{equation*}

	For the second claim, we note that Ref.~\onlinecite{lichong07} shows that the Ingleton inequality
	\eqref{ingleton} holds for the random variables $X_I = \pi_I(X)$.
	In the language of Ref.~\onlinecite{chanyeung02}, this is because the entropy vector $(H(X_I))$ can be characterized by
	the normal subgroups $\ker(\pi_I) \cap M^\perp$
	%, since $M^\perp / (\ker(\pi_I) \cap M^\perp) \cong M^{\perp_I}_I$
	(in fact, our phase spaces are even Abelian groups).
	Since the Ingleton inequality is balanced, the argument given above shows that it also holds for the von Neumann entropies of stabilizer states.
\end{proof}

Pure stabilizer states correspond to maximally isotropic submodules $M \subseteq V$.
Such submodules are called \emph{La\-grang\-ian}, and they satisfy $\abs M = d^n$ and $M = M^\perp$.
Thus in this case our classical model can also be defined by choosing $X \in M$ uniformly at random.
Furthermore, since $\pi_I(M) \cong M / (\ker \pi_I \cap M)$, we may also define $X_I$ to be the coset of $X$
modulo $\ker(\pi_I) \cap M = M \cap V_{I^c} = M_{I^c}$.
In this way, we recover the construction of Theorem 11 in Ref.~\onlinecite{linden_quantum_2013}.

\section{Gaussian States}
\label{gaussian states}

We sketch the corresponding result for Gaussian states of continuous-variable systems.
The Wigner function of an $n$-mode Gaussian quantum state $\rho$ with covariance matrix $\Sigma$ and first moments $\mu$ is defined as follows on classical phase space $\RR^{2n}$:
\begin{equation*}
	W_\rho(x) = \frac1{(2\pi)^{n} (\det\Sigma)^{\frac12}} e^{-\frac12 (x-\mu)^T \Sigma^{-1} (x-\mu)},
\end{equation*}
(see e.g.~the review Ref.~\onlinecite{weedbrockpirandolagarciapatronetal12}).
Evidently, $W_\rho$ is the probability density of a random vector $X = (X_1,\ldots,X_{2n})$ with multivariate normal distribution of mean $\mu$ and covariance matrix $\Sigma$. Using the well-known relation $\tr \rho^2 = (2\pi)^n \int W^2_\rho(x) dx$, it follows that the R\'{e}nyi-2 entropy of the quantum state, $S_2(\rho) = - \log \tr \rho^2$, is directly related to the differential R\'{e}nyi-2 entropy of the random variable $X$, $H_2(X) = - \log \int W^2_\rho(x)\ dx$:
\begin{equation}
\label{gaussian qc overall}
	S_2(\rho) = H_2(X) - n \log(2 \pi).
\end{equation}
The reduced state $\rho_I$ for some subset of modes $I \subseteq \{1,\ldots,n\}$ is again a Gaussian state, and its covariance matrix is equal to the corresponding submatrix of $\Sigma$. Thus the Wigner function of $\rho_I$ is given by the marginal probability density of the variables $X_I = (X_i)_{i\in I}$, and using \eqref{gaussian qc overall} we find that
\begin{equation}
\label{gaussian qc}
  S_2(\rho_I) = H_2(X_I) - \abs I \log(2 \pi).
\end{equation}
Equation \eqref{gaussian qc} states that the R\'{e}nyi-2 entropy of a Gaussian quantum state is always lower than the phase space entropy of its classical model, as given by the Wigner function. It is so by a precise amount, namely by $\log(2\pi)$ bits per mode.

\begin{theorem}[Classical model for Gaussian states]
	\label{gaussian theorem}
  Let $\rho$ be a Gaussian state with covariance matrix $\Sigma$, and
  define a random variable $X = (X_1,\ldots,X_n)$ with probability density
  given by the Wigner function $W_\rho(x)$. Then, for any positive $\alpha \neq 1$,
  \begin{equation*}
  	S_2(\rho_I) = H_\alpha(X_I) - \abs I \left( \log \pi -  \frac {\log \alpha} {1-\alpha} \right),
  \end{equation*}
  where $H_\alpha(X) = 1/(1-\alpha) \log \int W_\rho^\alpha(x)\ dx$ is the differential R\'{e}nyi-$\alpha$ entropy.
  In the limit $\alpha \rightarrow 1$, we recover
  \begin{equation}
  	\label{shannon gaussian}
	  S_2(\rho_I) = H(X_I) - \abs I \left( \log \pi + 1 \right).
  \end{equation}
  where $H(X) = - \int W_\rho(x) \log W_\rho(x)\ dx$ is the differential Shannon entropy.
\end{theorem}
\begin{proof}
	By Gaussian integration, the differential R\'{e}nyi-$\alpha$ entropy of the random variable $X_I$ is given by
	\begin{equation*}
		H_\alpha(X_I) = \frac 1 2 \log \det \Sigma + \abs I \left( \log 2\pi - \frac {\log \alpha} {1-\alpha} \right).
	\end{equation*}
	The assertions of the theorem follow from this and \eqref{gaussian qc}.
\end{proof}

Equation \eqref{shannon gaussian} has been previously used in Ref.~\onlinecite{adessogirolamiserafini12}, where the formula is attributed to Stratonovich.
Just as in the discrete case, we immediately get the following corollary:

\begin{corollary}
	\label{gaussian corollary}
	The R\'{e}nyi-2 entropy for Gaussian states satisfies all balanced information inequalities that are valid for multivariate normal distributions.
\end{corollary}

Interestingly, Gaussian states can violate the Ingleton inequality (as opposed to stabilizer states, cf.~Corollary \ref{stabilizer corollary}).
Indeed, this is well-known for multivariate normal distributions, and it is readily verified that the counterexample presented in Ref.~\onlinecite{shadbakhthassibi11} is a physical covariance matrix (i.e., it satisfies the \emph{uncertainty relation} $\Sigma + i \Omega \geq 0$, where $\Omega$ is the symplectic matrix).
Thus, by Theorem~\ref{gaussian theorem}, the corresponding Gaussian state violates the Ingleton inequality.

\section{Acknowledgements}

We would like to thank Matthias Christandl for many fruitful
discussions.
This work has benefited from insightful comments of an
anonymous referee and the associate editor.
DG's research is supported by the Excellence Initiative of the German
Federal and State Governments (ZUK 43). MW acknowledges support of the
Swiss National Science Foundation (PP00P2--128455), the German Science
Foundation (CH 843/2--1), and the National Center of Competence in Research `Quantum Science and Technology'.

\appendix
\section{Phase Space Approach to Stabilizer States}
\label{appendix}

In this appendix, we present a self-contained account of Weyl
operators and stabilizer states in the discrete phase-space picture.
This section does not contain original results. All statements could
be found in some form in
Refs.~\onlinecite{gottesman96,appleby05,gross06,beaudrap13,kueng13}, albeit
not in a unified language.

{\bf Discrete symplectic geometry.}
Let $d > 1$ be an integer and let $\ZZ_d=\ZZ/d\ZZ$ be the
congruence classes of integers modulo $d$.
The \emph{phase space} for $n$ particles with local dimension $d$ is by definition $V = \bigoplus_{i=1}^n V_i = \ZZ_d^{2n}$, the free $\ZZ_d$-module of rank $2n$. Given a point $v \in V$, we write $v_i = (p_i,q_i) \in V_i = \ZZ_d^2$ for its components. Consider the bilinear form $[-,-] \colon V \times V \rightarrow \ZZ_d$ defined by
$$
	[v,v'] = \sum_{i=1}^n p_i q_i' - q_i p_i'.
$$
It is non-degenerate and totally isotropic, i.e.\ $[v,v] = 0$ for all $v \in V$.
If $d$ is a prime then the phase space $V$ is simply a symplectic vector space over the finite field $\FF_d = \ZZ_d$. We will also in the general case refer to $[-,-]$ as the \emph{symplectic form}.

A \emph{character} of a finite Abelian group $G$ is a homomorphism $G \rightarrow U(1)$. Denote by $\widehat G$ the set of characters, which is again an Abelian group with the operation of pointwise multiplication. It is called the \emph{(Pontryagin) dual} of $G$. It is well-known that $G \cong \hat G$, although not canonically. For the cyclic group $G = \ZZ_d$, all characters are powers of $\chi_d(x) = e^{\frac{2\pi \mathrm{i}}{d} x}$.
%There is nothing canonical about this choice; any other primitive $d$-th root of unity would also induce an isomorphism.

\begin{lemma}
	\label{phase space characters}
	The characters of the additive group of the phase space $V$ are
 % all of the form $\chi_d([v,-])$, and every $v \in V$ gives a different character.
	$\hat V = \{ \chi_d([v,-]) : v \in V \} \cong V$.
\end{lemma}
\begin{proof}
	By injectivity of $\chi_d$ and non-degeneracy of the symplectic form, each $v$ determines a different character.
	Thus we have found all $\abs{\hat V} = \abs V$ many characters.
\end{proof}

The \emph{symplectic complement} of a submodule $M \subseteq V$ is the submodule $M^\perp = \{ v \in V : [v,m] = 0 \;\forall m \in M \}$.
In the case of prime $d$, it is well-known that $\dim M + \dim M^\perp = \dim V$---however, for general submodules the dimension (or rank) might not even be well-defined. Still there is an important analogue that holds in the general case:
\begin{lemma}
	\label{perp counting}
	$\abs M \abs{M^\perp} = \abs V$.
\end{lemma}
\begin{proof}
	We show that the group homomorphism
	\begin{equation*}
		\Phi \colon M^\perp \rightarrow \widehat{V/M},
		\quad
		v \mapsto \left( [w] \mapsto \chi_d([v,w]) \right).
	\end{equation*}
	is both injective and surjective (it is certainly well-defined).
	Injectivity follows immediately from the non-degeneracy of the symplectic form.
	For surjectivity, let $\tau \in \widehat{V/M}$. Then $w \mapsto \tau([w])$ is a character of $V$.
	By Lemma~\ref{phase space characters}, there exists $v \in V$ such that $\tau([w]) = \chi_d([v,w])$. Since $\tau$ vanishes on $M$, $v \in M^\perp$.
	Thus $\Phi$ is an isomorphism, and we find that
	\begin{equation*}
		\abs{M^\perp} = \abs{\widehat{V/M}} = \abs{V/M} = \frac {\abs V} {\abs M}.
		\qedhere
	\end{equation*}
\end{proof}
The following important corollary follows from Lemma~\ref{perp counting} and $M \subseteq (M^\perp)^\perp$.
\begin{corollary}
	\label{double perp}
  $(M^\perp)^\perp = M$.%, and $V/M^\perp \cong \widehat M$ via $[v] \mapsto \chi_d([v,-])$.
\end{corollary}
%For the second claim, consider the isomorphism $\Phi \colon M^perp \rightarrow \widehat{V / M}$ from the proof of \autoref{perp counting}.

We call a submodule $M \subseteq V$ an \emph{isotropic submodule} if $M \subseteq M^\perp$, i.e.~if $[m,m'] = 0$ for all $m, m' \in M$.

Finally, consider $V = \bigoplus_{i \in I} V_i$, the phase space of particles $I \subseteq \{ 1, \ldots, n\}$. There is a natural way of \emph{restricting} a submodule $M$ to $V_I$: we set
\begin{equation*}
 	M_I := M \cap V_I,
\end{equation*}
where $V_I$ is identified with a submodule of $V$ in the natural way.

\medskip

{\bf Weyl representation.}
Following Refs.~\onlinecite{appleby05,beaudrap13}, we first define
\emph{Weyl operators} for general integers $(P,Q) \in \ZZ^2$, not necessarily
in the range $\{0, \dots, d-1\}$. These are the unitaries
on $L^2(\ZZ_d) \cong \CC^d$ given by
\begin{equation*}
	(W({P,Q})\psi)(x) = \tau_{2d}(-P Q) \, \chi_d(P x) \, \psi(x - Q),
\end{equation*}
where $\tau_{2d}(R) = \chi_{2d}((d^2+1) R)$. % = (-1)^{d R} \chi_{2d}(R)$.
For example, $W({1,0})$ is the $Z$-operator $\ket x \mapsto e^{\frac
{2\pi\mathrm{i}} d x} \ket x$, while $W({0,1})$ is the $X$-operator $\ket
x \mapsto \ket{x+1 \pmod d}$.
By direct computation \cite{beaudrap13},
\begin{eqnarray}
	&& W({P,Q})W({P',Q'})\nonumber\\
	&=& \tau_{2d}(P Q'-Q P') \, W({P+P',Q+Q'}),
	\label{Z heisenberg} \\
	\nonumber \\
	&& W({P,Q})^{-1} \nonumber\\
	&=& W({P,Q})^\dagger = W({-P,-Q}), \label{Z adjoint} \\
	\nonumber \\
	&&W({P,Q}) W({P',Q'}) \nonumber\\
	&=& \chi_d(P Q'- Q P') \, W({P',Q'}) W({P,Q}).
	\label{Z commutator}
\end{eqnarray}

We now introduce the Weyl operators $w(p,q)$ for congruence classes
$(p,q)\in\ZZ_d^2$.  It is here that the treatment of the odd and the
even-dimensional case diverges.

For $d$ odd, $\tau_{2d}(1) = \chi_d\big(\frac{d^2 + 1} 2\big)$ is a $d$-th root of unity, so that $W(P+d, Q) = W(P, Q+d) = W(P,Q)$.
In other words, $W$ is constant on congruence classes modulo $d$, so
we can directly define $w(p,q) := W(P,Q)$.
In fact, $2^{-1}:=\frac{d^2+1}2\in\ZZ$ is
the multiplicative inverse of $2$ modulo $d$, so that we recover the
formulas from Section~\ref{stabilizer states}:
\begin{align}
	\nonumber
	(w(p, q) \psi)(x) = \chi_d(p x - 2^{-1} p q) \, \psi(x - q)	\\
	w(v) w(v') = \chi_d(2^{-1} [v,v']) \, w(v+v').
	\label{odd heisenberg}
\end{align}

For $d$ even, $\tau_{2d}(1) = \chi_{2d}(1)$ is a primitive $2d$-th
root of unity (e.g., in the case of qubits $\tau_{2d}(1)=i$).
Equation~\eqref{Z heisenberg} then implies that
$W(P+d,Q)$ and $W(P,Q+d)$ are either $W(P,Q)$ or $-W(P,Q)$. In order
to fix the sign, we choose $w(p,q):=W(P,Q)$, where
$(P,Q)$ is the unique preimage in $\{0,\ldots,d-1\}^2 \subseteq
\ZZ^2$.
Because $w$ and $W$ differ at most by a phase, \eqref{Z heisenberg}
still implies that $(p,q)\mapsto w(p,q)$ defines a projective
representation of the (additive structure of the) phase space
$\ZZ_d^2$.

In both the even and the odd case, it now follows
from \eqref{Z adjoint} and \eqref{Z commutator}
that
\begin{align}
	\label{adjoint}
	w(v)^{-1} &= w(v)^\dagger = \pm w(-v), \\
	\label{commutator}
	w(v) w(v') &= \chi_d([v,v']) \, w(v') w(v).
\end{align}
%(The $\pm$ is due to the fact that the preimage of $-v$ is in general not equal to the negative preimage of $v$.)
%(Qubit example: $i \sigma_y$ has order 4 but all elements in $\ZZ_2^2$ have order two.)

For $n$-particles, the phase space is the direct sum $V = \bigoplus_{i=1}^n V_i = \ZZ_d^{2n}$. We define its Weyl representation on $(\CC^d)^{\otimes n}$ by the tensor product of the single-particle representations, $w(v) = \bigotimes_{i=1}^n w(v_i)$. In this way, the relations \eqref{adjoint} and \eqref{commutator} continue to hold. Moreover, it is easy to verify that
\begin{equation}
	\label{trace}
	\tr w(v) = d^n \, \delta_{v,0}.
\end{equation}

\medskip

{\bf Stabilizer states.}
To define stabilizer states, we start with a \emph{stabilizer group}
$G$, i.e.\ a finite Abelian group whose elements are multiplies of
Weyl operators on $(\CC^d)^{\otimes n}$, such that the only multiple
of $\Id=w(0)$ contained in $G$ is $\Id$ itself. With such a group we
associate the operator
\begin{equation*}
	P = \frac 1 {\abs G} \sum_{g \in G} g.
\end{equation*}
From the fact that $G$ is a group, we deduce that $P^2=P$;
%Equation~\eqref{Z adjoint} give
since all elements of $G$ are unitaries, $P=P^\dagger$;
and \eqref{trace} implies that $\tr
P=d^n/|G|$. Hence $P$ projects onto a $\big(d^n/|G|\big)$-dimensional
subspace, called the \emph{stabilizer code} of $G$.
The corresponding \emph{stabilizer state} is $\rho = \frac 1 {d^n} \sum_g g$.
We now prove Theorem~\ref{stabilizers}, which we repeat for the reader's convenience:

\stabilizertheorem{stabilizers}{}{\tag{\ref{reductions stabilizers}}}{\tag{\ref{von neumann entropy stabilizers}}}
\begin{proof}
  \emph{1.\ From isotropic submodules to classes of stabilizer
	states:}
	Let $M \subseteq \ZZ_d^{2n}$ be an isotropic submodule.
	Since $M$ is a finite Abelian group, it can be written as a direct sum of cyclic groups,
	$M \cong \ZZ_{d_1} \oplus \ldots \oplus \ZZ_{d_k}$.
	Let $m_j \in M$ be a generator of the $j$-th cyclic subgroup.
	Since $w(m_j)^{d_j} \propto w(0) = \mathbf 1$, we can choose phases $\lambda_j$ such that $(\lambda_j w(m_j))^{d_j} = \mathbf 1$.
	Define
	\begin{equation*}
		G = \{ \underbrace{\prod_{j=1}^k (\lambda_j w(m_j))^{x_j}}_{=:
		\mu_m \, w(m)} : m = \sum_j x_j m_j \in M \}.
	\end{equation*}
	The product is well-defined, because by \eqref{commutator}, the Weyl
	operators $\{ w(m) : m \in M \}$ all commute. Thus the data $M, \mu$
	define a
	stabilizer group of cardinality $\abs M$ with
	corresponding stabilizer state
	$
		\rho(M, \mu) = \frac 1 {d^n} \sum_{m \in M} \mu_m w(m).
	$
	This state depends on the phases $\mu_m$, which in turn resulted from
	our choice of generators $m_j$ and phases $\lambda_j$. A different
	choice would have resulted in another stabilizer group $\{\nu_m
	w(m) \,|\, m\in M\}$ and
	we have yet to show that the two groups are related by conjugation
	with some Weyl operator.
	To this end, note that \eqref{trace} implies that necessarily $\nu_m w(m) \nu_{m'} w(m') = \nu_{m+m'} w(m+m')$.
	It follows that $\tau(m) := \nu_m / \mu_m$ defines a character of $M$.
%	It is well-known that for $M \subseteq V$, the Pontryagon dual $\widehat V$ surjects onto $\widehat M$.
	Since $V$ is an Abelian group, this character can be extended to all of $V$, and it is therefore by Lemma \ref{phase space characters} of the form
	$\tau(m) = \chi_d([v,m])$ for some $v \in V$. But then it follows from \eqref{commutator} that
	\begin{equation*}
		w(v) \, \mu_m w(m) \, w(v)^\dagger
		= \tau(m) \mu_m w(m)
		= \nu_m w(m).
	\end{equation*}
	Consequently, $\rho(M, \nu)$ and $\rho(M, \mu)$ are related by
	conjugation with the Weyl operator $w(v)$.

	If $d$ is odd, then \eqref{odd heisenberg} implies that $w(m) w(m')
	= w(m+m')$. It follows that $G := \{ w(m) : m \in M \}$ is a
	stabilizer group of cardinality $\abs M$, with corresponding
	stabilizer state
	$
		\rho(M) = \frac 1 {d^n} \sum_{m \in M} w(m).
	$
	This is the canonical	representative \eqref{stabilizer state} of
	the equivalence class of states associated with $M$.

	\emph{2. Surjectivity:}
	Here, we show that
	our map from isotropic submodules to equivalence classes of
	stabilizer states is surjective, Let $G$ be a stabilizer group with
	corresponding stabilizer state $\rho$.  Equation \eqref{trace}
	implies that, for each $g \in G$ there exists a unique $m_g \in V$
	such that $g \propto w(m_g)$.  Conversely, no two $m_g$ can be
	equal---
	otherwise, two group elements in $G$ would differ only by a phase
	and hence there would be a non-trivial multiple of $\Id$ in $G$.
	Define $M	:= \{ m_g \}$. Then $M$ is a submodule of $V$, since
	%$$w(m_{gh}) \propto gh \propto w(m_g) w(m_h) \propto w(m_g + m_h)$$ implies that
	$m_g + m_h = m_{gh}$.
	Since $G$ is Abelian, all $w(m_g)$ commute and \eqref{commutator} shows that $M$ is isotropic.
	Then $M$ is indeed a preimage of $[\rho]$, since by its very construction there exists a choice of phases by which we recover $G$ (namely $\mu_{m_g} = g \, w(m_g)^{-1}$).

	\emph{3. Injectivity:}
	Suppose that $\rho(M,\mu)$ and $\rho(M',\mu')$ are two equivalent stabilizer states. As we saw, conjugating with a Weyl operator only changes the phases, so we may in fact assume that states are equal. Now assume that $M \neq M'$, so that there exists e.g.~$m \in M \setminus M'$. Then, \eqref{trace} shows that
	\begin{equation*}
		0 \neq \tr w(m) \rho(M,\mu) = \tr w(m) \rho(M',\mu') = 0,
	\end{equation*}
	which is the desired contradiction.

	\emph{4. Reduction:}
	We now show that our construction is compatible with reduction. For this, observe that
	\begin{equation*}
		\tr_{I^c} w(m) =
		%\prod_{i \in I} w(m_i) \prod_{j \in I^c} \tr w(j) =
		\left( \prod_{i \in I} w(m_i) \right) d^{\abs I^c} \, \delta_{m \in M_I}.
	\end{equation*}
	Since any valid assignment of phases $\mu_m$ restricts to the submodule $M_I = M \cap V_I$,
	it follows that $[\rho(M)_I] = [\rho(M_I)]$.
	It is also immediate that the canonical element \eqref{stabilizer state} is compatible with reduction.

	\emph{5. Entropy:}
	In view of the last point, it suffices to show the statement about entropies for $I = \{1,\ldots,n\}$.
	Recall that the cardinality of $M$ and of the corresponding stabilizer groups $G$ agree.
	We have already seen that the dimension of the stabilizer code is equal to $d^n / \abs G$.
	Thus,
	\begin{equation*}
		S([\rho(M)]) = n - \log \, \abs G = n - \log \, \abs M. \qedhere
	\end{equation*}
\end{proof}

\end{document}